\documentclass[conference]{IEEEtran}
\usepackage{spconf,amsmath,amssymb,graphicx,epsfig, cite,algorithm,algorithmic,epstopdf, url}
\usepackage[utf8]{inputenc}
\usepackage[mathscr]{euscript}
\usepackage{color}
\usepackage{bm}
\usepackage{amsthm}

\def\minwrt[#1]{\underset{#1}{\text{minimize }}}
\def\argminwrt[#1]{\underset{#1}{\text{arg min }}}
\def\argmaxwrt[#1]{\underset{#1}{\text{arg max }}}
\def\maxwrt[#1]{\underset{#1}{\text{maximize }}}
\def\maxemphwrt[#1]{\underset{#1}{\text{\emph{maximize} }}}

\newtheorem{theorem}{Theorem}
\newtheorem{remark}{Remark}
\newtheorem{proposition}{Proposition}
\newtheorem{lemma}{Lemma}
\newtheorem{definition}{Definition}

\newcommand{\abs}[1]{\left|#1\right|}

\def\by{{\bf y}}

\def\RN{{\mathbb{N}}}

\newcommand{\misparam}{\theta}

\newcommand{\expop}{\mathbb{E}}

\newcommand{\inharm}{\Delta} 
\newcommand{\missignal}{\mu} 
\newcommand{\signal}{x} 
\newcommand{\noiseysignal}{y} 
\newcommand{\misnoise}{w} 
\newcommand{\noise}{e} 
\newcommand{\wavediff}{\xi} 
\newcommand{\pitch}{\omega} 
\newcommand{\truefreq}{\tilde{\omega}} 
\newcommand{\amp}{r} 
\newcommand{\trueamp}{\tilde{r}} 
\newcommand{\phase}{\phi} 
\newcommand{\truephase}{\tilde{\phi}} 
\newcommand{\var}{\sigma^2} 
\newcommand{\truevar}{\tilde{\sigma}^2} 
\newcommand{\misscale}{\alpha}

\newcommand{\phasediff}{\breve{\phi}}
\newcommand{\freqdiff}{\breve{\omega}}

\newcommand{\fim}{F}
\newcommand{\extrahessian}{\tilde{F}}

\newcommand{\rreal}{\mathfrak{R}}
\newcommand{\iimag}{\mathfrak{I}}
\newcommand{\firstelement}{\eta}

\newcommand{\signalpdf}{f}
\newcommand{\modellikelihood}{\mathcal{L}}

\newcommand{\mcrlb}{\mathrm{MCRLB}}

        \makeatletter
        \def\fps@eqnfloat{!t}
        \def\ftype@eqnfloat{4}
        
        \newenvironment{eqnfloat*}
               {\@dblfloat{eqnfloat}}
               {\end@dblfloat}
        \makeatother
\renewenvironment{thebibliography}[1]{%
 \begin{oldthebibliography}{#1}%
 \setlength{\parskip}{0ex}%
 \setlength{\itemsep}{0ex}%
}%
{%
\end{oldthebibliography}%
}%
\begin{document}
\title{On Harmonic Approximations of Inharmonic Signals}
\name{Filip Elvander$^*$, Jie Ding$^{*,\dagger}$, Andreas Jakobsson$^*$\thanks{This work was supported in part by the Swedish Research Council, Carl Trygger's foundation, and the Royal Physiographic Society in Lund.}}

\address{
$^*$Div. of Mathematical Statistics, Lund University, Sweden\\
$^\dagger$College of Underwater Acoustic Engineering, Harbin Engineering University, China\\
emails: \{filip.elvander, andreas.jakobsson\}@matstat.lu.se,  dingjie@hrbeu.edu.cn\vspace{-4pt}}

\maketitle
\begin{abstract}
In this work, we present the misspecified Gaussian Cram\'er-Rao lower bound for the parameters of a harmonic signal, or pitch, when signal measurements are collected from an almost, but not quite, harmonic model. For the asymptotic case of large sample sizes, we present a closed-form expression for the bound corresponding to the pseduo-true fundamental frequency. Using simulation studies, it is shown that the bound is sharp and is attained by maximum likelihood estimators derived under the misspecified harmonic assumption. It is shown that misspecified harmonic models achieve a lower mean squared error than correctly specified unstructured models for moderately inharmonic signals. Examining voices from a speech database, we conclude that human speech belongs to this class of signals, verifying that the use of a harmonic model for voiced speech is preferable.
%
%
\end{abstract}
\vspace{2mm}
%
\section{Introduction}
Signals displaying harmonic structures arise in a wide set of applications, ranging from speech processing \cite{NorholmJC16_24} to machinery fault detection \cite{Randall11}, with the fundamental frequency, or pitch \cite{ChristensenJ09}, often being used as a characterizing feature for the signal \cite{LittleMHSR09_56}. The problem of finding statistically efficient, as well as computationally feasible \cite{NielsenJJCJ17_135}, estimators of the fundamental frequency constitutes an active field of research, including recent efforts for addressing multi-pitch signals \cite{ElvanderSJ17_25,AdalbjornssonJC15_109}. The assumption underlying the derivation of such methods is that of perfect harmonicity, i.e., the frequencies of the sinusoids constituting each pitch group should be exact integer multiples of a corresponding fundamental \cite{ChristensenJ09}. However, for some signal sources, this harmonic relationship is only approximate, i.e., the sinusoidal frequencies may deviate slightly from the postulated harmonic structure. Such discrepancy, referred to as inharmonicity, is an inherent property of the sound produced by stringed musical instruments, being caused by the stiffness of the vibrating strings \cite{Fletcher62_36}. Also, the voiced part of human speech often displays some inharmonicity, albeit with no apparent structure \cite{GeorgeS97_5}. Although methods for estimating the parameters of inharmonic signals have been proposed, either by exploiting parametric models \cite{ZhangCJM10_18,ButtASJ13_icassp} or by using robust distance measures \cite{ElvanderAKJ17_icassp}, little attention has been directed to analyzing the achievable estimation accuracy when applying estimators derived under the perfectly harmonic assumption. That is, how well, in terms of variance and mean squared error, may the parameters of an inharmonic signal be estimated by harmonic estimators? This work aims to address this question by considering this problem within the framework of misspecified estimation (see, e.g., \cite{FortunatiGGR17_34} for an overview). This framework allows for finding the so-called pseudo-true parameters of an assumed model, which are the expected values of unbiased estimators derived under the assumed model, when applied to actual signal measurements. In particular, this allows for defining a pseudo-true value of the fundamental frequency of a harmonic signal, even though this does not exist in a strictly physical sense. Furthermore, a lower bound on the variance of misspecified estimators may be found by considering an extension of the Cram\'er-Rao lower bound (CRLB), referred to as the misspecified CRLB (MCRLB) \cite{RichmondH15_63}. Taken together, this allows for finding a lower bound on the mean squared error of pitch estimators when applied to inharmonic signals, yielding a means of quantifying the loss of performance caused by the presence of inharmonicity.
Under the assumption of signal observations in additive Gaussian noise, we here determine the pseudo-true parameters of misspecified harmonic models when used for approximating inharmonic measurements. We also present the MCRLB for the corresponding parameters and show that the bound for the pseudo-true fundamental frequency can be found in closed form asymptotically. The theoretical findings are validated using numerical simulations, showing that the derived bounds are sharp, i.e., attained by the maximum likelihood estimator (MLE) derived under the harmonic assumption. The theoretical results support the heuristic that harmonic approximations of slightly inharmonic signals are not only computationally preferable, but also statistically superior to unstructured but exact models.
Using the Keele Pitch Reference database \cite{PlanteMA95_eurospeech}, we evaluate the typical inharmonicity found in human speech, concluding that the deviations are generally sufficiently small so that it is preferable to exploit the harmonic structure, even if imperfect, when estimating the pitch.

\section{Signal model and misspecification}
Consider the measured signal\footnote{Here, for generality, we will consider the complex-valued representation, noting that this can easily be formed as the discrete-time analytical version of a real-valued signal  \cite{Marple99_47}.}
\begin{align}\label{eq:sine_model}
	y_t = \signal_t + e_t= \sum_{k=1}^K \trueamp_k e^{i\truephase_k + i\truefreq_k t} + \noise_t,
\end{align}
for $t = 0,1,\ldots,N-1$, for $N\in \RN$, where $\trueamp_k>0$, $\truephase_k \in [0,2\pi)$, $\truefreq_k \in [0,2\pi)$, and $\noise_t$ is a circularly symmetric white Gaussian noise with variance $\truevar$. Further, assume that the sinusoidal frequencies $\truefreq_k$ satisfy
\begin{align}
	\truefreq_k = \pitch k + \inharm_k
\end{align}
for some $\pitch \in [0,2\pi)$. The offsets $\inharm_k$ are referred to as the inharmonicity parameters, as, for the case $\inharm_k = 0$, for $k= 1,\ldots,K$, the signal is perfectly harmonic. This type of quasi-periodic structure is observed in many forms of signals, such as the voiced part of human speech \cite{GeorgeS97_5} and in the sound produced by stringed musical instruments. For the latter case, a commonly utilized model to describe $\truefreq_k$ is
\begin{align}\label{eq:piano_model}
	\truefreq_k = \pitch k\sqrt{1+\beta k^2},
\end{align}
where $\beta \geq 0$ is the string stiffness parameter \cite{Fletcher62_36}. Works on finding high accuracy estimates of $\left\{ \truefreq_k \right\}_{k=1}^K$ by incorporating knowledge of the almost harmonic signal structure in \eqref{eq:sine_model} typically rely on parametric models for the sinusoidal frequencies, such as \eqref{eq:piano_model} \cite{ChristensenJ09}.
In such settings, potential gains are achieved by using fewer non-linear parameters than the $K$ needed for an unstructured sinusoidal model. In contrast, we herein seek to analyze and quantify the loss in estimation accuracy incurred by assuming a rigid harmonic model. That is, we seek to approximate the signal in \eqref{eq:sine_model} by
\begin{align}\label{eq:pitch_model}
	\noiseysignal_t = \missignal_t + \misnoise_t= \sum_{k=1}^K \amp_k e^{i\phase_k + ik \pitch t} + \misnoise_t,
\end{align}
where $\misnoise_t$ is a circularly symmetric white Gaussian noise with variance $\var$, and where $\pitch$ is the fundamental frequency. As may be noted that, from an implementation point of view, 
this model is preferable as it has only one non-linear parameter, i.e., the fundamental frequency $\pitch$. However, this simplicity may be expected to come at a price. For example, if considered estimates of $\left\{ \truefreq_k \right\}_{k=1}^K$, the sequence $\left\{ k\hat{\pitch} \right\}_{k=1}^K$ may be biased for estimators $\hat{\pitch}$ that are unbiased under \eqref{eq:pitch_model}. Also, lower bounds on the variance of estimators $\hat{\pitch}$ may be expected to be different from bounds corresponding to the model \eqref{eq:pitch_model} such as, e.g., the CRLB. In the following section, we aim to address these issues by considering the concept of the pseudo-true parameters of \eqref{eq:pitch_model} and their misspecified CRLB (MCRLB).
%
%
\section{Pseudo-true parameters and performance bounds}
Although the model \eqref{eq:sine_model} does not have a fundamental frequency $\pitch$, one may still define such a concept through the use of pseudo-true parameters. Specifically, letting the parameters of the assumed harmonic model be
\begin{align}
	\misparam = \left[\begin{array}{ccccccc} \pitch & \phase_1 & \ldots & \phase_K & \amp_1 & \ldots & \amp_K\end{array}\right]^T,
\end{align}
one may consider the following definition.
%
\begin{definition}[Pseudo-true parameter\cite{FortunatiGGR17_34}] 
	Consider a signal sample $\by$ with probability density function $\signalpdf$. For a likelihood $\modellikelihood$, parametrized by the parameter vector $\misparam$, the pseudo-true parameter, $\theta_0$, is defined as
	\begin{align} \label{eq:pseudo_true}
		\misparam_0 = \argminwrt[\misparam] -\mathbb{E}_{\signalpdf}\left( \log {\modellikelihood}(\by;\misparam)  \right).
	\end{align}
\end{definition}
%
Here, it may be noted that the pseudo-true parameter $\misparam_0$ minimizes the Kullback-Leibler divergence between the distribution of the actual measurement, i.e., $\signalpdf$, and the distribution of the model, encoded in the parametric likelihood $\modellikelihood$. Interestingly, it can be shown that the MLE corresponding to the misspecified model $\modellikelihood$ converges to the pseudo-true parameter \cite{FortunatiGGR17_34}. That is, as the number of samples from \eqref{eq:sine_model} tend to infinity, the maximumm likelihood estimator (MLE) derived under \eqref{eq:pitch_model} tends to the pseudo-true parameter, making this an applicable definition of fundamental frequency for practical purposes.
For the case of estimating the parameters of \eqref{eq:pitch_model} from measurements from \eqref{eq:sine_model}, the following proposition holds.
%
%
\begin{proposition}[Pseudo-true parameter]
The pseudo-true parameter for the pitch model in \eqref{eq:pitch_model} is given by
\begin{align} \label{eq:pseudo_true_param}
	\misparam_0 = \argminwrt[\misparam] \sum_{t=0}^{N-1} \abs{\signal_t - \missignal_t(\misparam)}^2,
\end{align}
and the pseudo-true variance is given by
\begin{align}
	\var = \truevar + \frac{1}{N} \sum_{t=0}^{N-1} \abs{\wavediff_t(\misparam_0)}^2
\end{align}
where $\wavediff_t(\misparam) = \missignal_t(\misparam) - x_t$ is the expected difference in waveform.
\end{proposition}
\begin{proof}
As both the assumed and true distributions are Gaussian, the result follows directly.
\end{proof}
%
Thus, the pseudo-true parameter vector $\misparam_0$ minimizes the $\ell_2$-distance between the assumed and actual signal waveforms. It may be further noted that $\misparam_0$ is unique for finite $N$. With this, we may conclude that the expected biases for estimates of the sinusoidal frequencies are $\left\{ k\pitch_0 - \truefreq_k \right\}_{k=1}^K$. However, in order to find the mean squared error (MSE) for such estimators, we also require the estimator variance. A lower bound on this variance is given by the MCRLB. Specifically, the following theorem from \cite{RichmondH15_63} holds. 
%
%
%
\begin{theorem} \label{th:mrclb}
Let $\hat{\misparam}$ be an estimator of $\misparam_0$ that is unbiased under $\signalpdf$. Then,
\begin{align} \label{eq:mrclb}
	\expop_\signalpdf\left( (\hat{\misparam}-\misparam_0)(\hat{\misparam}-\misparam_0)^T  \right) \succeq \truevar A(\misparam_0)^{-1}\fim(\misparam_0) A(\misparam_0)^{-1}
\end{align}
where
\begin{align*}
	\fim(\misparam) \!=\!\frac{2\truevar}{(\var)^2}\!\sum_{t=0}^{N-1}\!\nabla_\misparam \missignal_t^\rreal(\misparam) \nabla_\misparam \missignal_t^\rreal(\misparam)^T\!+\!\nabla_\misparam \missignal_t^\iimag(\misparam) \nabla_\misparam \missignal_t^\iimag(\misparam)^T
\end{align*}
and $A(\misparam) = -\frac{\var}{\truevar}\fim(\misparam) - \extrahessian(\misparam)$, with
\begin{align*}
	\extrahessian(\misparam)\!= \!\frac{2}{\var}\sum_{t=0}^{N-1} \left(\wavediff_t^\rreal(\misparam) \nabla^2_\misparam \missignal_t^\rreal(\misparam)  + \wavediff_t^\iimag(\misparam) \nabla^2_\misparam \missignal_t^\iimag(\misparam) \right).
\end{align*}
Here, $(\cdot)^\rreal$ and $(\cdot)^\iimag$ denote the real and imaginary parts, respectively. 
The proof may be found in \cite{RichmondH15_63}.
\end{theorem}
%
%
The MCRLB is given by the diagonal of the right-hand side of \eqref{eq:mrclb} and thus provides a lower bound on the variance of any estimator of $\theta_0$ that is unbiased under $\signalpdf$. It may be noted that for the case of $\inharm_k = 0 $, for $k = 1,\ldots,K$, i.e., when the signal in \eqref{eq:sine_model} is perfectly harmonic, $\extrahessian(\misparam) = 0$, $A(\misparam_0) = -\fim(\misparam_0)$, and the MCRLB coincides with the CRLB of a harmonic signal. Further, as $N \to \infty$, one may express the MCRLB corresponding to the pseudo-true fundamental frequency in closed form, 
as detailed below. 
%
%
%
\begin{figure}[t!]
        \centering
            \includegraphics[width=0.45\textwidth]{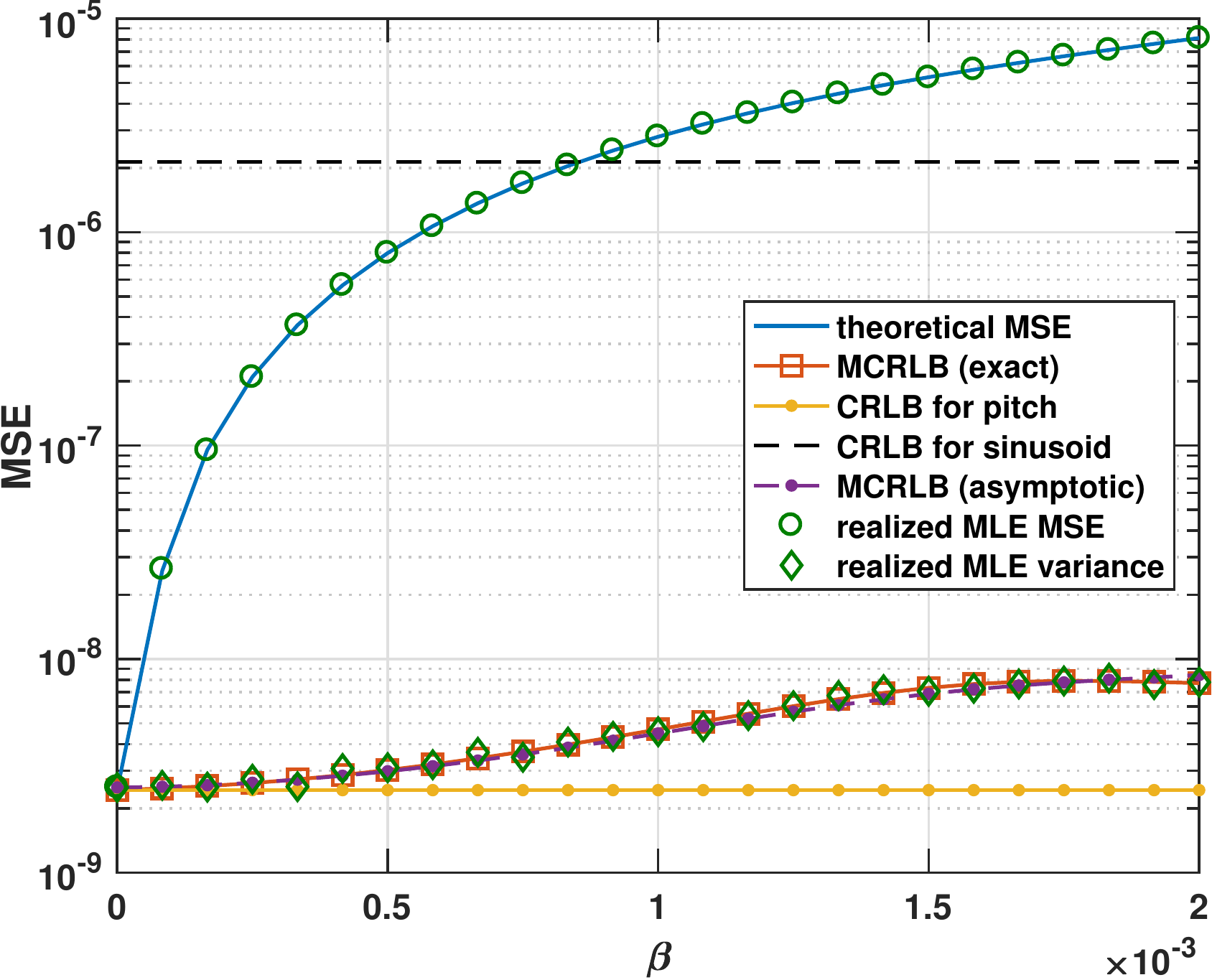}
           \caption{Mean squared error when varying the inharmonicity parameter $\beta$, for SNR 10dB and $N = 200$.}
            \label{fig:mse_vary_beta}
\end{figure}
%
\begin{proposition}[Asymptotic MCRLB] \label{prop:asymp_mcrlb}
Let the pseudo-true parameter be
\begin{align*}
	\misparam_0 = \left[\begin{array}{ccccccc} \pitch_0 & \phase_1 & \ldots & \phase_K & \amp_1 & \ldots & \amp_K\end{array}\right]^T.
\end{align*}
Then, as $N\to \infty$, the asymptotic $\mcrlb$ for the pseudo-true fundamental frequency $\pitch_0$ is given by
\begin{align}
	\mcrlb(\pitch_0) = \truevar \frac{C + E}{\left(C-E + Z + D   \right)^2}
\end{align}
where $C = \frac{N(N^2-1)\sum_{k=1}^K k^2 \amp_k^2}{6}$, and
\begin{align*}
	Z &= -2 \sum_{k=1}^K k^2\amp_k^2 \frac{N(N-1)(2N-1)}{6} \\
	&\quad + 2 \sum_{k=1}^K \sum_{t=0}^{N-1} k^2\amp_k\trueamp_k t^2 \cos(\phasediff_k+\freqdiff_kt) \\
	\!D\!&=\!2(N\!-\!1)\!\Bigg[\!\frac{N(\!N\!-\!1\!)}{2}\!\sum_{k=1}^K\!k^2\amp_k^2\!-\!\sum_{k=1}^K\!\sum_{t=0}^{N-1} \!k^2\amp_k\trueamp_k\cos(\phasediff_k\!+\!\freqdiff_k t)\!\Bigg]\\
	E &= \frac{2}{N}\sum_{k=1}^K k^2 \trueamp_k^2 \left( \sum_{t=0}^{N-1} t \sin(\phasediff_k+\freqdiff_kt) \right)^2 \\ 
	&\quad+ \frac{2}{N}\sum_{k=1}^K k^2 \left( \trueamp_k \sum_{t=0}^{N-1} t \cos(\phasediff_k+\freqdiff_kt)  - \amp_k\frac{N(N-1)}{2}\right)^2
\end{align*}
where $\phasediff_k = \phase_k - \truephase_k$ and $\freqdiff_k = k\pitch_0 - \truefreq_k$, for $k = 1,\ldots,K$.
\end{proposition}
\begin{proof}
	See appendix.
\end{proof}
%
%
%
\begin{figure}[t!]
        \centering
            \includegraphics[width=0.45\textwidth]{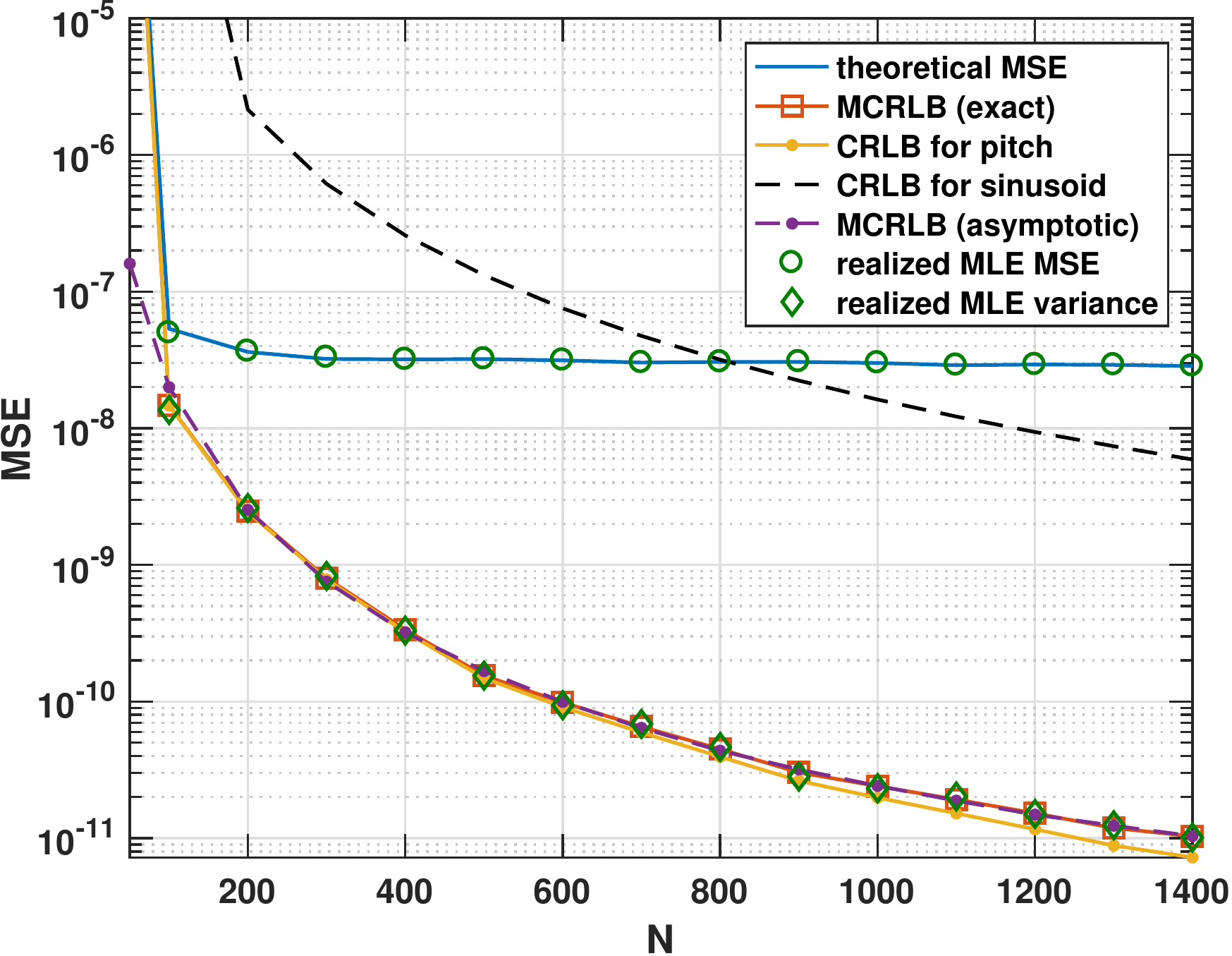}
           \caption{Mean squared error when varying the number of signal samples $N$, for $\beta = 10^{-4}$ and SNR 10dB.}
            \label{fig:mse_vary_N}
\end{figure}
%
\begin{remark}
It may here be noted that $\truevar/C$ corresponds to the asymptotic CRLB case for a pitch model in a correctly specified setting \cite{ChristensenJJ07_15}. It may also be noted that $Z = D = E = 0$ when the pseudo-true and true parameters coincide.
\end{remark}
With this, we may compute (asymptotic) lower bounds on the MSE for estimators of the sinusoidal frequencies $\truefreq_k$, derived under the assumption of the harmonic model \eqref{eq:pitch_model}. Specifically, these lower bounds are given by
\begin{align} \label{eq:mse}
	\expop_\signalpdf\left( (k\hat{\pitch} - \truefreq_k)^2 \right) \geq (k\pitch_0 - \truefreq_k)^2 + k^2\mcrlb(\pitch_0)
\end{align}
for $k = 1,\ldots,K$. In a practical estimation scenario, this is the expected performance obtained when applying harmonic estimators to almost harmonic signals \eqref{eq:sine_model}.
%
%
%
%
\begin{figure}[t!]
        \centering
            \includegraphics[width=0.45\textwidth]{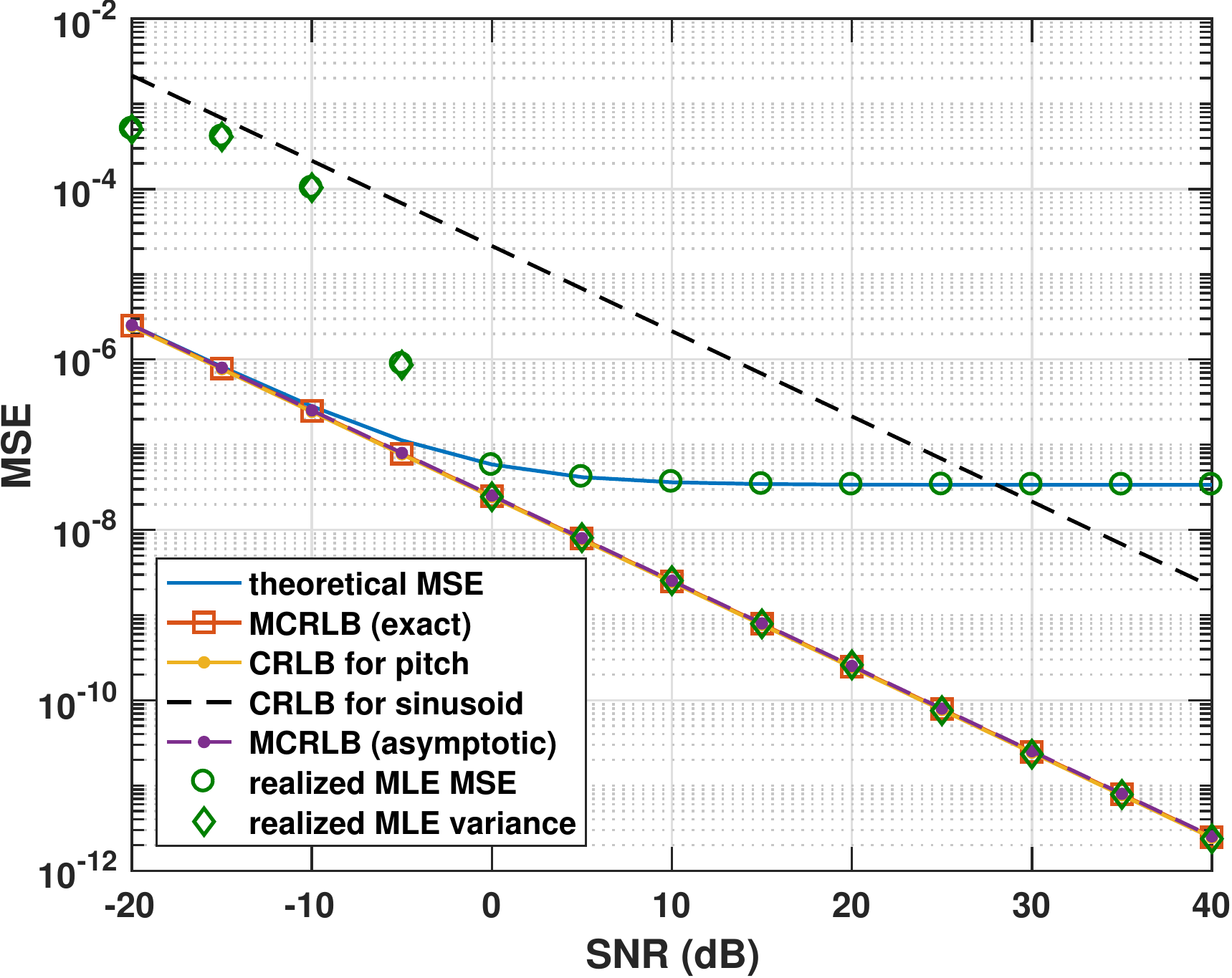}
            \vspace{-2mm}
           \caption{Mean squared error when varying the SNR, for $\beta = 10^{-4}$ and $N = 200$.}
            \label{fig:mse_vary_SNR}
\vspace{-4mm}
\end{figure}
%
%
%
%
%
%
%
%
%
%
\section{Numerical results}
\vspace{-1mm}
In this section, we consider the MSE and variance obtained for the estimate of the lowest sinusoidal frequency when applying the MLE \cite{ChristensenJ09} of the parameters of \eqref{eq:pitch_model} to measurements from \eqref{eq:sine_model}, for varying degrees of inharmonicity, sample sizes, and signal-to-noise-ratio (SNR). Here, SNR is defined as $\text{SNR} = \sum_{k=1}^K \trueamp_k^2/\truevar$.
For simplicity, we use the model in \eqref{eq:piano_model} for the frequencies. Thus, $\beta = 0$ corresponds to a perfectly harmonic model. We set $K = 10$, $\omega = \pi/40$, and use the sinusoidal amplitudes $\trueamp_k = e^{-\frac{1}{20}(k-K/2)^2}$, for $k = 1,\ldots,K$. The initial phases are chosen uniformly at random in $[0,2\pi)$. For each considered setting, i.e., for a certain $\beta$, $N$, and SNR, we conduct 1000 Monte Carlo simulations, from which the estimator MSE and variance for the lowest order sinusoid are estimated. Also, for reference, the CRLB for the unstructured model in \eqref{eq:sine_model}, as well as for the perfectly harmonic model, are provided.
Figure~\ref{fig:mse_vary_beta} shows the bound as a function of the inharmonicity parameter $\beta$. As may be seen, the MCRLB is orders of magnitude smaller than the CRLB for the sinusoidal model, even for large values of $\beta$. Also, for small values of $\beta$, the MSE if lower than the CRLB for the unstructured model, indicating that the misspecified model is expected to perform better for small deviations from the harmonic model. Figures~\ref{fig:mse_vary_N} and \ref{fig:mse_vary_SNR} considers varying sample size $N$ and SNR, respectively. As can be seen, the harmonic model here is expected to perform better in terms of MSE unless the number of samples or the SNR is large. The conclusion from this is quite intuitive; in adverse estimation scenarios, i.e., when the sample size is small or the noise level is high, one is expected to gain from exploiting the signal structure, even if it is only approximate. Conversely, for large number of samples and low noise levels, one is better off using the exact signal characteristics. As may be noted, the loss of performance for the misspecified model is caused by the systematic bias; the variance is considerably lower than for the sinusoidal model in all considered cases. It may also be noted that the asymptotic expression for the MCRLB corresponds well to the exact values, also for quite large inharmonicities.
%
\begin{figure}[t!]
        \centering
            \includegraphics[width=0.45\textwidth]{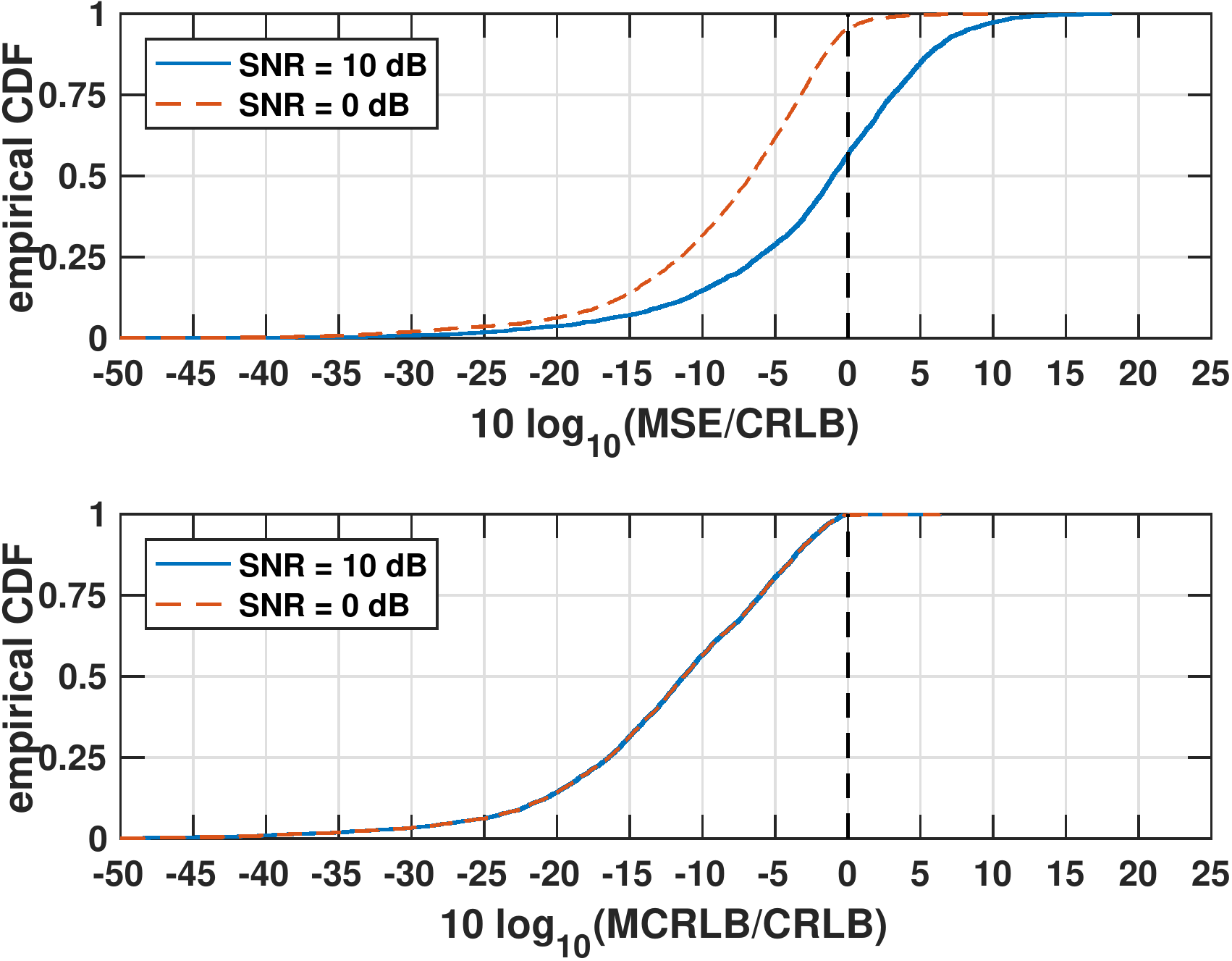}
            \vspace{-1mm}
           \caption{Empirical distribution function for ratio between the MSE and MCRLB for the potentially misspecified harmonic model and the CRLB for the unstructured sinusoidal model.}
            \label{fig:empirical_cdf}
\vspace{-5mm}
\end{figure}
%
In order to get an idea of the impact of these results on real data, we consider recordings of human speech from the Keele Pitch Reference database \cite{PlanteMA95_eurospeech}. 
We subdivide the voiced part of the recordings into frames of length 25.6 ms, and for each frame estimate the parameters of the sinusoidal model in \eqref{eq:sine_model} using the corresponding MLE. The obtained parameters are then used to compute the corresponding pseduo-true parameters, $\misparam_0$, as well as the MCRLB and theoretical MSE for the pseduo-true fundamental frequency. We then compare these quantities to the CRLB for the lowest sinusoidal frequency under the unstructured model in \eqref{eq:sine_model}. All bounds are computed for SNRs 0 and 10 dB, assuming  $N = 200$ signal samples. To avoid ambiguities, only signal frames with 3-10 sinusoidal components without missing harmonics are included. This results in sinusoidal parameters being estimated from 3655 signal frames, corresponding to 22\% of the total number of frames labeled as voiced. The empirical cumulative distribution functions (CDFs) for the resulting ratios are shown in Figure~\ref{fig:empirical_cdf}. As can be seen, the MSE for the (potentially) misspecified harmonic model is lower than the CRLB for the sinusoidal model in approximately 57\% and 90\% of the cases, for SNRs 10 and 0 dB, respectively, whereas the MCRLB is smaller for virtually all cases. It may be noted the the relative magnitude of the MCRLB and the sinusoidal CRLB is invariant under scaling of the noise power. This supports the heuristic of using harmonic models for potentially inharmonic measurements for voice data. It should be stressed that Figure~\ref{fig:empirical_cdf} is conservative in the sense of being constructed to the benefit of the sinusoidal model in \eqref{eq:sine_model}, with estimation errors propagation as to exaggerating the MSE of the misspecified inharmonic model.
\appendix
Here, we provide a proof of Proposition~\ref{prop:asymp_mcrlb}.
%
\begin{proof}
By Lemma~\ref{eq:lemma_arrowhead} below, $\frac{1}{N} \fim(\misparam_0)$ and $\frac{1}{N} A(\misparam_0)$ converge to arrowhead matrices. The structure for $\frac{1}{N} A(\misparam_0)$ is
\begin{align}
	\frac{1}{N} A(\misparam_0) = -\frac{1}{N\var} \begin{bmatrix}
		\firstelement & z^T \\
		z & \text{diag}(d)
	\end{bmatrix}
\end{align}
where $\firstelement = \firstelement_\missignal+ \firstelement_\wavediff$ and $z = z_\missignal + z_\wavediff$ with
\begin{align}
	\firstelement_\missignal&= 2\sum_{t=0}^{N-1} \left( \frac{\partial \missignal_t^\rreal }{\partial \pitch}  \right)^2 +  \left( \frac{\partial \missignal_t^\iimag }{\partial \pitch}  \right)^2  \\
	\firstelement_\wavediff &= 2\sum_{t=0}^{N-1} \wavediff_t^\rreal \frac{\partial^2 \missignal_t^\rreal}{\partial \pitch^2} + \wavediff_t^\iimag\frac{\partial^2 \missignal_t^\iimag}{\partial \pitch^2} \\
	d &= 2\sum_{t=0}^{N-1}\nabla_\misscale \missignal_t^\rreal \odot \nabla_\misscale \missignal_t^\rreal +\nabla_\misscale \missignal_t^\iimag \odot \nabla_\misscale \missignal_t^\iimag \\
	z_\missignal&= 2\sum_{t=0}^{N-1}\nabla_\misscale \missignal_t^\rreal \frac{\partial \missignal_t^\rreal}{\partial \pitch} + \nabla_\misscale \missignal_t^\iimag \frac{\partial \missignal_t^\iimag}{\partial \pitch}\\
	z_\wavediff &= 2\sum_{t=0}^{N-1}\wavediff^\rreal_t \frac{\partial}{\partial \pitch}\nabla_\misscale \missignal_t^\rreal + \wavediff_t^\iimag \frac{\partial}{\partial \pitch}\nabla_\misscale \missignal_t^\iimag,
\end{align}
where $\odot$ denotes the Hadamard product, and
\[
\alpha = \left[ \begin{array}{cccccc} \phase_1 & \ldots & \phase_K & \amp_1 & \ldots & \amp_K \end{array} \right]^T,
\]
with all derivatives being evaluated at $\misparam = \misparam_0$.
The inverse of $\frac{1}{N} A(\misparam_0)$ can then be written, using the Sherman-Morrison-Woodbury formula \cite{GolubV96}, as
\begin{align}
	\left( \frac{1}{N} A(\misparam_0) \right)^{-1} = -\var\begin{bmatrix}
		0 & 0\\ 0 & N\text{diag}(d)^{-1}
	\end{bmatrix} - \frac{N\var}{\rho}uu^T
\end{align}
where $u = u_\missignal + u_\wavediff$, with
\begin{align}
	 \rho &= \firstelement - z^T(z./ d) \\
	 u_\missignal &= \left[ \begin{array}{cc} -1 & (z_\missignal./ d)^T \end{array} \right]^T \\
	 u_\wavediff &= \left[ \begin{array}{cc} 0 & (z_\wavediff./ d)^T \end{array} \right]^T,
\end{align}
where $./$ denotes elementwise division, implying that $A(\misparam_0)^{-1}$ converges to
\begin{align}
	A(\misparam_0)^{-1} = -\var\begin{bmatrix}
		0 & 0 \\ 0 & \text{diag}(d)^{-1}
	\end{bmatrix} - \frac{\var}{\rho}uu^T.
\end{align}
As may be noted, the MCRLB corresponding to $\pitch$ is given by the first diagonal element of $A(\misparam)^{-1}\fim(\misparam)A(\misparam)^{-1}$. This element is given by the first element of the matrix
\begin{align}
	(\var)^2\frac{1}{\rho}u u^T F u \frac{1}{\rho}u^T = (\var)^2\left(\frac{1}{\rho^2} u^TFu \right) uu^T,
\end{align}
which, as the first element of $u u^T$ is 1, is $(\var)^2\frac{1}{\rho^2} u^TFu$. In the same sense, $ \fim(\misparam_0)$ converges to
\begin{align}
	\fim(\misparam_0) = \frac{\truevar}{(\var)^2}\begin{bmatrix}
		\firstelement_\missignal & z_\missignal^T \\
		z_\missignal & \text{diag}(d)
	\end{bmatrix}
\end{align}
It is readily verified that $u_\missignal^T \fim u_\wavediff = 0$, yielding
\begin{align}
	\frac{(\var)^2}{\rho^2} u^TFu &= (\var)^2\frac{1}{\rho^2} \left( u_\missignal^TFu_\missignal + u_\wavediff^TFu_\wavediff \right) \\
	&= \frac{\truevar}{\rho^2}\left( \firstelement_\missignal - z_\missignal^T(z_\missignal./d) + z_\wavediff^T(z_\wavediff./d)   \right)\\
	&= \truevar\frac{\firstelement_\missignal - z_\missignal^T(z_\missignal./d) + z_\wavediff^T(z_\wavediff./d) }{\rho^2}.
\end{align}
Furthermore, noting that
\begin{align*}
	\rho = \firstelement_\missignal - z_\missignal^T(z_\missignal./d) -  z_\wavediff^T(z_\wavediff./d)-2z_\missignal^T(z_\wavediff./d) +  \firstelement_\wavediff,
\end{align*}
we may write
\begin{align}
	\frac{(\var)^2}{\rho^2} u^TFu = \var \frac{C + E}{\left( C - E + Z + D \right)^2}
\end{align}
where
\begin{align}
	C &= \firstelement_\missignal - z_\missignal^T(z_\missignal./d) \\
	E &= z_\wavediff^T(z_\wavediff./d)\\
	D &= -2z_\missignal^T(z_\wavediff./d) \\
	Z &= \firstelement_\wavediff.
\end{align}
Assuming that the pseudo-true fundamental frequency is not too close to zero, the correlation between signal components corresponding to different harmonic orders tends to zero as $N\to \infty$. The asymptotic expressions for $C, E, D, $ and $Z$ stated in the proposition follow directly.
\end{proof}
%
\begin{lemma} \label{eq:lemma_arrowhead}
As $N\to \infty$, $\frac{1}{N} \fim(\misparam_0)$ and $\frac{1}{N} A(\misparam_0)$ converge to arrowhead matrices.
\end{lemma}
%
\begin{proof}
Firstly, it may be noted that as $\misparam_0$ solves the least squares criterion in \eqref{eq:pseudo_true_param}, it directly follows from the optimality criterion that
\begin{align}
	\sum_{t=0}^{N-1} \wavediff_t^\rreal \nabla_\misparam \missignal_t^\rreal + \sum_{t=0}^{N-1} \wavediff_t^\iimag \nabla_\misparam \missignal_t^\iimag = 0.
\end{align}
Then, as any second derivative of $\missignal_t^\rreal$ and $\missignal_t^\iimag$ not involving differentiation with respect to $\pitch$ is equal to a constant real scaling, i.e., not depending on $t$, of a corresponding element of $\nabla_\misparam \missignal_t^\rreal$ and $\nabla_\misparam \missignal_t^\iimag$, respectively, it follows that
\begin{align}
	\sum_{t=0}^{N-1} \wavediff_t^\rreal \nabla_\misscale^2 \missignal_t^\rreal + \sum_{t=0}^{N-1} \wavediff_t^\iimag \nabla_\misscale^2 \missignal_t^\iimag = 0,
\end{align}
when all quantities are evaluated at $\misparam = \misparam_0$. Thus, only elements of $\extrahessian(\misparam_0)$ related to partial derivatives with respect to $\pitch$ are non-zero, and we may conclude that only the first column and first row of $\extrahessian(\misparam_0)$ are non-zero, which holds for any $N \in \RN$. Considering the elements of $\fim(\misparam_0)$, it may be noted that for elements not containing partial derivatives with respect to $\pitch$,
\begin{align*}
	\sum_{t=0}^{N-1} \frac{\partial \missignal_t^\rreal}{\partial \amp_\ell}\frac{\partial \missignal_t^\rreal}{\partial \amp_k} +\frac{\partial \missignal_t^\iimag}{\partial \amp_\ell}\frac{\partial \missignal_t^\iimag}{\partial \amp_k} &= \sum_{t=0}^{N-1} \cos(\phase_\ell-\phase_k + (\ell-k)\pitch t) \\
	\!\sum_{t=0}^{N-1}\!\frac{\partial \missignal_t^\rreal}{\partial \phase_\ell}\frac{\partial \missignal_t^\rreal}{\partial \phase_k}\!+\!\frac{\partial \missignal_t^\iimag}{\partial \phase_\ell}\frac{\partial \missignal_t^\iimag}{\partial \phase_k}\!&=\!\amp_\ell \amp_k\!\sum_{t=0}^{N-1}  \cos(\phase_\ell\!-\!\phase_k\!+\!(\ell\!-\!k)\pitch t) \\
	\!\sum_{t=0}^{N-1}\!\frac{\partial \missignal_t^\rreal}{\partial \phase_\ell}\frac{\partial \missignal_t^\rreal}{\partial \amp_k}\!+\!\frac{\partial \missignal_t^\iimag}{\partial \phase_\ell}\frac{\partial \missignal_t^\iimag}{\partial \amp_k}\!&=\!\sum_{t=0}^{N-1}\!\amp_\ell  \sin(\phase_\ell\!-\!\phase_k\!+\!(\ell\!-\!k)\pitch t).
\end{align*}
From this, one may conclude that any such off-diagonal element converges to zero when normalized by $\frac{1}{N}$, whereas the diagonals are constants identical to either $1$ or $r_k^2$, when scaled in the same way. That is, for large $N$, these off-diagonal elements are negligible compared to the diagonal. Using the same line of reasoning, it can be shown that the off-diagonal elements related to partial derivates of $\pitch$ grow linearly when scaled by $\frac{1}{N}$, whereas the first element on the diagonal grows quadratically. Thus, $\frac{1}{N} \fim(\misparam_0)$, and thereby $\frac{1}{N}A(\misparam_0)$, converges to an arrowhead matrix as $N \to \infty$.
\end{proof}
%
%
\section{References}
\centering{\normalsize{\textbf{REFERENCES}}}
\vspace{1mm}
\bibliographystyle{IEEEbib}
\bibliography{IEEEabrv,mis_specified_pitch_icassp2020_arxiv2.bbl}
%
\end{document}